\documentclass[a4paper,UKenglish,cleveref, autoref, thm-restate,authorcolumns]{lipics-v2019}

\usepackage{complexity}
\newcommand{\ccP}{\mathcal{P}}
\newcommand{\LE}{\text{LE}}

\newcommand{\jump}{\mathrm{jump}}
\newcommand{\bump}{\mathrm{bump}}
\usepackage{upgreek}
\usepackage{algorithm,bbm}
\usepackage[noend]{algpseudocode}

\title{Exact exponential algorithms for two poset problems} %

\titlerunning{Exact exponential algorithms for two poset problems} %

\author{L\'{a}szl\'{o} Kozma}{Freie Universit\"at Berlin, Institute of Computer Science, Germany}{laszlo.kozma@fu-berlin.de}{}{Research supported by DFG grant KO 6140/1-1.}%

\authorrunning{L. Kozma} %

\Copyright{L\'{a}szl\'{o} Kozma} %



\category{} %

\supplement{}%

\nolinenumbers %

\EventEditors{~}
\EventNoEds{1}
\EventLongTitle{~}
\EventShortTitle{arXiv 2020}
\EventAcronym{~}
\EventYear{~}
\EventDate{~}
\EventLocation{~}
\EventLogo{}
\SeriesVolume{~}
\ArticleNo{0}

\begin{document}

\maketitle

\begin{abstract}
Partially ordered sets (posets) are fundamental combinatorial objects with important applications in computer science. Perhaps the most natural algorithmic task, given a size-$n$ poset, is to compute its number of \emph{linear extensions}. In 1991 Brightwell and Winkler showed this problem to be $\#\P$-hard. In spite of extensive research, the fastest known algorithm is still the straightforward $O(n 2^n)$-time dynamic programming (an adaptation of the Bellman-Held-Karp algorithm for the TSP). Very recently, Dittmer and Pak showed that the problem remains $\#\P$-hard for \emph{two-dimensional} posets, and no algorithm was known to break the $2^n$-barrier even in this special case.  
The question of whether the two-dimensional problem is easier than the general case was raised decades ago by M\"ohring, Felsner and Wernisch, and others. In this paper we show that the number of linear extensions of a two-dimensional poset can be computed in time $O(1.8172^n)$. 

The related \emph{jump number problem} asks for a linear extension of a poset, minimizing the number of neighboring \emph{incomparable} pairs. The problem has applications in scheduling, and has been widely studied. In 1981 Pulleyblank showed it to be \NP-complete. We show that the jump number problem can be solved (in arbitrary posets) in time $O(1.824^n)$. This improves (slightly) the previous best bound of Kratsch and Kratsch. 
\end{abstract}

\section{Introduction}
\label{sec:introduction}
A \emph{partially ordered set} (\emph{poset}) $\ccP = (X,\prec)$ consists of a ground set $X$ and an {irreflexive} and {transitive} binary relation $\prec$ on $X$. 
A \emph{linear extension} of $\ccP$ is a total order on $X$ that contains $\prec$. The main problem considered in this paper is to determine, given a poset $\ccP$ on a ground set of size $n$, the number of linear extensions $\LE(\ccP)$ of $\ccP$. We refer to this counting problem as \#LE. A poset can alternatively be seen as a transitive directed acyclic graph (DAG), where \#LE asks for the number of topological orderings of the graph.

Posets are fundamental objects in combinatorics (for a detailed treatment we refer to the monographs~\cite{Trotter_book, Rival_book}, \cite[§\,3]{Stanley_book}, \cite[§\,8]{Handbook_comb}) with several applications in computer science. For instance, every comparison-based algorithm (e.g.\ for sorting) implicitly defines a sequence of posets on the input elements, where each poset captures the pairwise comparisons known to the algorithm at a given time. 
An efficient sorter must find comparisons whose outcomes split the number of linear extensions in a balanced way. A central and long-standing open question in this area is whether a comparison with ratio (at worst) $1/3:2/3$ exists in every poset~\cite{Brightwell_survey}; slightly weaker constant ratios are known to be achievable~\cite{KahnSaks,Brightwell1995}.

Counting linear extensions (exactly or approximately) is a bottleneck in experimental work, e.g.\ when testing combinatorial conjectures. In computer science the \#LE problem is relevant, besides the mentioned task of optimal comparison-based sorting,  
for learning graphical models~\cite{Wallace, Niinimaki}, probabilistic ranking~\cite{Winkler_avg, Fishburn, Lukasiewicz}, reconstruction of partial orders from sequential data~\cite{MannilaMeek}, convex rank tests~\cite{Morton}, multimedia delivery in networks~\cite{Amer}, and others.

The complexity of \#LE has been thoroughly studied (see Linial~\cite{Linial} for an early reference). Lov\'asz~\cite[§\,2.4]{Lovasz} mentions the problem as a special case of polytope volume computation; Stanley~\cite{Stanley_polytope} gives a broad overview of the polytope-formulation of \#LE. Brightwell and Winkler~\cite{Brightwell1991} show that \#LE is $\# \P$-hard, and thus unlikely to admit a polynomial-time solution. In fact, despite the significant attention the problem has received (e.g.\ the mentioned papers and references therein and thereof), the best upper bound on the running time remains $O(n2^n)$. This bound can be achieved via dynamic programming over the subsets of the ground set~\cite{KSZ, DeLoof}, an approach\footnote{A finer bound on the running time is $O(w \cdot |I|)$, where $w$ is the \emph{width} of the poset, and $I$ is its set of \emph{ideals} (i.e.\ downsets); in the worst case, however, this expression does not improve the given bound. } that closely resembles the Bellman-Held-Karp algorithm for the \emph{traveling salesman problem} (\emph{TSP})~\cite{Bellman1962,HeldKarp1962}.

A bound of $2^n$ appears to be a natural barrier\footnote{We only study exact algorithms in this paper; for \emph{approximating} $\LE(\ccP)$, fully polynomial-time randomized schemes are known~\cite{Dyer, Bubley}.} for the running time of \#LE, similarly to some of the most prominent combinatorial optimization problems (e.g.\ \emph{set cover}/\emph{hitting set}, \emph{CNF-SAT}, \emph{graph coloring}, \emph{TSP}).\footnote{The \emph{strong exponential time hypothesis}~\cite{IPZ} states that a running time $O(c^n)$ with $c<2$ is not achievable for CNF-SAT, and a similar barrier has been conjectured for set cover~\cite{SeCoCo}.}  
We show that \#LE can be solved faster when the input poset is \emph{two-dimensional}. 
Dimension is perhaps the most natural complexity-measure of posets, and can be seen informally as a measure of the \emph{nonlinearity} of a poset (see e.g.\ Trotter~\cite{Trotter_book}). 
As one-dimensional posets are simply total orders, the first nontrivial case is dimension two. The structure of two-dimensional posets is, however, far from trivial. Posets in this class capture the \emph{point-domination order} in the plane.  

The question of the complexity of \#LE in two-dimensional posets was raised in the 1980s by M\"ohring~\cite{Mohring} and later by Felsner and Wernisch~\cite{FelsnerWernisch}. An even earlier mention of the problem is by Atkinson, Habib, and Urrutia, in a discussion of open problems concerning posets, cf.\ Rival~\cite[p.\ 481]{Rival_book}.

Efficient algorithms for \#LE are known for various restricted classes of posets, e.g.\ \emph{series-parallel}~\cite{Mohring}, \emph{low treewidth}~\cite{LE_tw, LE_sparse, LE_treewidth3}, \emph{small width}~\cite{DeLoof}, avoiding certain \emph{substructures}~\cite{LE_nfree}, and others; see M\"ohring~\cite{Mohring} for an early survey of tractable special cases. However, as the techniques used in these works rely on certain kinds of \emph{sparsity} in the input, they are not applicable for the case of two-dimensional posets. It is easy to see that the latter may be arbitrarily dense, containing, for example, a complete bipartite graph of linear size. 
In fact, Dittmer and Pak~\cite{DittmerPak} recently showed that \#LE is $\#\P$-hard already for this class of inputs. 
Our first result is stated in the following theorem.

\begin{theorem}\label{thm1}
The number of linear extensions of a two-dimensional poset of size $n$ can be computed in time $O(1.8172^n)$. 
\end{theorem}

Our second result is an algorithm for the \emph{jump number problem}. In this (optimization) problem a linear extension of $\ccP$ is sought, such as to minimize the number of adjacent pairs of elements that are incomparable in $\ccP$ (such pairs are called \emph{jumps}). The problem is known to be \NP-hard~\cite{Pulleyblank}, and has been well-studied due to its applications in scheduling. 

Similarly to \#LE, the jump number problem can be solved by dynamic programming in time $2^n n^{O(1)}$. An improved algorithm with running time $O(1.8638^n)$ was given by Kratsch and Kratsch~\cite{jumpnr}. We also refer to their paper for further background and motivation for the problem. Improving the bound of Kratsch and Kratsch, we obtain the following result. 

\begin{theorem}\label{thm3}
The jump number problem can be solved in time $O(1.824^n)$. 
\end{theorem}

Note that in this case no assumption is made on the dimension of the input poset.  
Whether \emph{jump number} remains \NP-hard in two-dimensional posets is a long-standing open question~\cite{Mohring, Ceroi, Steiner87}.

\subparagraph*{Poset dimension.}
Formally, the dimension $\dim(\ccP)$ of a poset $\ccP = (X,\prec)$ is the smallest number $d$ of total orders, whose intersection is $\ccP$. In other words, if $\dim(\ccP)=d$, then there exists a collection of orders $<_1, \dots, <_d$ (called \emph{realizers} of $\ccP$), such that for all $x,y \in X$, we have $x \prec y$ if and only if $x <_k y$ for all $1 \leq k \leq d$. 

Poset dimension was introduced by Dushnik and Miller in 1941~\cite{DM}, and the concept has since been extensively studied; we refer to the monograph of Trotter dedicated to poset dimension theory~\cite{Trotter_book}. Various kinds of \emph{sparsity} of $\ccP$ are known to imply upper bounds on $\dim(\ccP)$ (see e.g.\ \cite{Joret2018, ScottWood} for recent results in a long line of such works). The converse is, in general, not true, as two-dimensional posets may already be arbitrarily dense, and are known not to have a characterisation in terms of finitely many forbidden substructures~\cite{Baker1970, Mohring}.

The term \emph{dimension} is motivated by the following natural geometric interpretation. Suppose $\ccP$ is a $d$-dimensional, size-$n$ poset with realizers $<_1, \dots, <_d$. The ground set can then be viewed as a set of $n$ points in $d$-dimensional Euclidean space, with no two points aligned on any coordinate, such that the ordering of the points according to the $k$-th coordinate coincides with the order $<_k$, for all $1 \leq k \leq d$. The partial order $\prec$ is then exactly the \emph{point-domination order}, i.e.\ $x \prec y$ if and only if all $d$ coordinates of $y$ are larger than the corresponding coordinates of $x$. 
In this geometric view,  
a linear extension of a low-dimensional poset can be seen as a tour that visits all points, never moving behind the \emph{Pareto front} of the already visited points.

Two-dimensional posets are particularly natural, as they are in bijection with permutations (the ranks of points by $<_1$ and $<_2$ can be seen respectively as the \emph{index} and \emph{value} of a permutation-entry). Swapping the two coordinates yields a \emph{dual} poset, turning chains into antichains and vice versa. It follows that the complement of the \emph{comparability graph} is itself a comparability graph, which is yet another exact characterization of two-dimensional posets.\footnote{Given a poset $\ccP = (X,\prec)$, its \emph{comparability graph} is $\C(\ccP) = (X,E)$, where $\{x,y\} \in E$ if $x \prec y$ or $y \prec x$. 
The \emph{width} of $\ccP$ is the size of the largest \emph{antichain} in $\ccP$, i.e.\ independent set in $\C(\ccP)$, and the \emph{height} of $\ccP$ is one less than the size of the largest \emph{chain} in $\ccP$, i.e.\ clique in $\C(\ccP)$. 
} It is not hard to see that two-dimensional posets are exactly the \emph{inclusion-posets} of intervals on a line.

Yet another interpretation of two-dimensional posets relates them to the weak Bruhat order on permutations. In this setting the number of linear extensions of a two-dimensional poset equals the number of permutations that are \emph{reachable} from a given permutation $\pi$ by a sequence of swaps between mis-sorted adjacent elements; a question of independent interest~\cite{FelsnerWernisch, DittmerPak}.

\section{Counting linear extensions in two-dimensional posets}
\label{sec2}

Denote $[k] = \{1,\dots,k\}$. For a set $Y$ with partial order $\prec$, let $\max{(Y)}$ denote the \emph{set of maxima} of $Y$, i.e.\ the set of elements $x \in Y$ with the property that $x \prec y$ implies $y \notin Y$.

Let $\ccP = (X,\prec)$ be a size-$n$ poset. To introduce the main elements of our \#LE algorithm, we review first the classical $O(n 2^n)$ time algorithm. 

For all $Y \subseteq X$, let $\LE(Y)$ denote the number of linear extensions of the subposet of $\ccP$ induced by $Y$, and let $\LE(\emptyset) = 1$. We recursively express $\LE(Y)$ for all nonempty $Y$, by removing in turn all elements that can appear at the end of a total order on $Y$: 
\begin{align}
\LE(Y) = \sum_{x \in \max{(Y)} } {\LE\left(Y \setminus \{x\}\right)}.\label{rec1}
\end{align}

To compute $\LE(\ccP) = \LE(X)$, we evaluate recurrence~(\ref{rec1}), saving all intermediate entries $\LE(Y)$ for $Y \subseteq X$. There are at most $2^n$ such entries, and computing each takes $O(n)$ time, once the results of the recursive calls are available. (With simple bookkeeping, $\max{(Y)}$ is available for all calls without additional overhead.)

\subsection{A first improvement}

A well-known observation is that when computing $\LE(X)$ by (\ref{rec1}), only those subproblems $Y \subseteq X$ arise where $y \in Y$ and $x \prec y$ imply $x \in Y$, i.e.\ the \emph{downsets} of $\ccP$. In general, the number of downsets can be as high as $2^n$, when $\ccP$ consists of a single antichain. Nonetheless, we can give better bounds on the number of downsets, if necessary, by modifying the input poset $\ccP$.

\subparagraph*{Large matching case.}

An observation already made in previous works (e.g.\ \cite{jumpnr}) is the following. Consider a size-$m$ \emph{matching} $M$ in the comparability graph $\C(\ccP)$, with matched edges $\{x_i,y_i\}$, where $x_i \prec y_i$, for all $i \in [m]$. Let $W$ denote the set of vertices matched by $M$ and let $A = X \setminus W$.

Then, the sets $Y \subseteq X$ where $Y \cap \{x_i,y_i\} = \{y_i\}$ for some $i \in [m]$ are not downsets and cannot be reached by recursive calls. The remaining sets can be partitioned as $T_0 \cup T_1 \cup \dots \cup T_m$, where $T_0 \subseteq A$ is an arbitrary subset of the unmatched vertices, and $T_i \in \{\emptyset, \{x_i\},\{x_i,y_i\}\}$ for $i \in [m]$. %

The number of sets of this form is %
$2^{n-2m} \cdot 3^m$. If $m = \upalpha n$, this quantity equals 
$(2 \cdot (\frac{3}{4})^\upalpha)^n$.
When $\upalpha \geq \frac{1}{3}$, the number of subproblems is thus less than $1.8172^n$, and the running time is within the required bounds.

\subparagraph*{Small matching case.}

Let us assume from now on that $M$ is a \emph{maximum} matching of size $m = \upalpha n$ for $\upalpha < \frac{1}{3}$. %
The maximality of $M$ implies that the unmatched vertices $A$ form an independent set in $\C(\ccP)$, i.e.\ an antichain of $\ccP$, of size $|A| = (1-2\upalpha) n$. We assume $\upalpha > 0$, as otherwise $\ccP$ is a single antichain and the problem is trivial.

For $x \in A$, let $N(x)$ denote the \emph{open neighborhood} of $x$ in $\C(\ccP)$, i.e.\ the set of elements in $X$ that are comparable with $x$. Observe that $N(x) \cap A = \emptyset$ for all $x \in A$. 

If $N(x) \cap A = \emptyset$ for an element $x \in W$, then we say that $x$ is \emph{incomparable} with $A$.
Otherwise, if $x \prec y$ for some $y \in N(x) \cap A$, we say that $x$ is \emph{below} $A$, and if $y \prec x$, for some $y \in N(x) \cap A$, we say that $x$ is \emph{above} $A$.
Observe that $x$ cannot be both below and above $A$, as that would make two elements of $A$ comparable, contradicting the fact that $A$ is an antichain.

The sets $N(x)$ define a \emph{partition} of $A$, where $x,y \in A$ are in the same class if and only if $N(x) = N(y)$. In general posets there can be as many as $2^{n-|A|}$ classes. The following lemma states that in two-dimensional posets the number of classes is much smaller.

\begin{lemma}\label{lem:struct2}
Let $\ccP = (X,\prec)$ be a size-$n$ poset, with $\dim(\ccP) \leq 2$, and let $A \subseteq X$ be an antichain. Then, $N(\cdot)$ partitions $A$ into at most $2(n-|A|)$ classes. 
\end{lemma}   

Before proving Lemma~\ref{lem:struct2}, we show that it can be used to compute $\LE(\ccP)$ more efficiently. 
Let $A_1, \dots, A_\ell$ be the partition of $A$ defined by $N(\cdot)$, and for each $i \in [\ell]$, denote $a_i = |A_i|$. 
Let $x_i^k$, where $i \in [\ell]$ and $k \in [a_i]$, be a \emph{virtual element}, and let $Q$ denote the set of all such virtual elements.

Construct a new poset $\ccP' = (X',\prec')$ as follows. Let $X' = W \cup Q$. In words, the ground set $X'$ contains all vertices matched by $M$, and instead of the elements of the antichain $A$, it contains the virtual elements of $Q$. Observe that $|Q|=|A|$ and therefore $|X'|=|X|$.

The relation $\prec'$ is defined as follows, covering all cases:
\begin{itemize}
\item if $x,y \in W$, then $x \prec' y$ $\iff$ $x \prec y$,
\item if $x = x_i^p$ and $y = x_j^q$, then $x \prec' y$ $\iff$ $i=j$ and $p<q$,
\item if $x \in W$ and $y = x_i^p$, then $x \prec' y$ $\iff$ $x \prec z$, for some $z \in A_i$,
\item if $x = x_i^p$ and $y \in W$, then $x \prec' y$ $\iff$ $z \prec y$, for some $z \in A_i$.
\end{itemize}

In words, $\prec'$ preserves the relation $\prec$ between elements of $W$. Virtual elements with the same index $i$ form a chain $x_i^1 \prec' \cdots \prec' x_i^{a_i}$, for all $i \in [\ell]$. Virtual elements with different indices are incomparable. The relation between a virtual element $x_i^k$ and an element $y \in W$ preserves the relation $\prec$ between an arbitrary element $z \in A_i$ and $y$. The choice of $z$ is indeed arbitrary, as the elements in $A_i$ are by definition indistinguishable. 

Intuitively, $x_i^k$ is a placeholder for the element of $A_i$ that appears as the $k$-th among all elements of $A_i$ in some linear extension of $\ccP$. The sequence $x_i^1, \dots, x_i^{a_i}$ corresponds to an arbitrary permutation of the elements of $A_i$. This intuition is captured by the following statement.

\begin{lemma} With the above definitions:
$$\LE(\ccP) = \LE(\ccP') \cdot \prod_{i \in [\ell]}{(a_i!)} .$$ \label{lem:struct3}
\vspace{-0.2in}
\end{lemma}

Let us postpone proving Lemma~\ref{lem:struct3} as well, and state our first algorithm, \#LE-2D, as Algorithm~\ref{alg1}. The algorithm constructs the poset $\ccP'$ and computes its number of linear extensions using recurrence (\ref{rec1}), then computes the correct count for $\ccP$ via Lemma~\ref{lem:struct3}.

\begin{algorithm}
  \caption{Algorithm~\#LE-2D}\label{alg1}
  \begin{algorithmic}[1]
    \Statex \textbf{Input:} Poset $\ccP = (X,\prec)$, where $|X| = n$.
    \Statex \textbf{Output:} The number of linear extensions $\LE(\ccP)$ of $\ccP$. 
    \State Find a maximum matching $M$ of $\C(\ccP)$ with vertex set $W$. 
    \State Let $A = X \setminus W$.
    \State Let $A_1, \dots, A_\ell$ be the partition of $A$ by the neighborhoods in $\C(\ccP)$.
    \State Let $a_i = |A_i|$ for $i \in [\ell]$.
    \State Construct $\ccP' = (X',\prec')$, as described.
    \State Compute $N = \LE(\ccP')$  using (\ref{rec1}). %
    \State \textbf{return} $N \cdot \prod_{i\in[\ell]}{(a_i!)}$.
    \end{algorithmic}
\end{algorithm}

\subparagraph*{Analysis of the running time.} 
Step 1 amounts to running a standard maximum matching algorithm (see e.g.\ \cite{Tarjan_book}). 
Computing the partition in Step 3 takes linear time with careful data structuring. Steps 2,4,5,7 clearly take linear time overall. 

The polynomial-time overhead of steps other then Step 6, as well as the polynomial factor in the analysis of (\ref{rec1}) are absorbed in the exponential running time of Step 6, where we round the base of the exponential upwards. To derive a worst-case upper bound on the running time of Step 6, it only remains to bound the number of downsets of $\ccP'$.

Observe that the ground set $X'$ can, by construction, be partitioned into chains. The matched vertices of $W$ form $m$ chains $x_i \prec' y_i$, for $i \in [m]$, as before. The virtual elements of $Q$ form $\ell$ chains of lengths $a_1, \dots, a_\ell$, where the $i$-th chain is $x_i^1 \prec' \cdots \prec' x_i^{a_i}$. All downsets of $\ccP'$ are then of the form $(T_1 \cup \dots \cup T_m) \cup (Q_1 \cup \dots \cup Q_{\ell})$, where $T_i \in \{\emptyset, \{x_i\},\{x_i,y_i\}\}$ for $i \in [m]$, and $Q_i = \{x_i^j : j \leq t_i\}$, for some threshold $0 \leq t_i \leq a_i$, for $i \in [\ell]$. %

The number of such sets is $3^m \cdot \prod_{i \in [\ell]}{(a_i+1)}$. 
Recall that $m = \upalpha n$ and $|A| = (1-2\upalpha)n$. The quantity $\prod_{i=1}^\ell{(a_i + 1)}$ is maximized when the values $a_i+1$ are all equal, and thus equal to $\frac{|A| + \ell}{\ell}$, yielding the overall upper bound $3^{\upalpha n} \left(\frac{(1-2\upalpha)n}{\ell}+1\right)^\ell$. (Observe that $\ell \geq 1$ always holds.)

Since the quantity is increasing in $\ell$, and $\ell \leq 2(n - |A|) =  4\upalpha n$ by Lemma~\ref{lem:struct2}, we obtain the upper bound $\left( 3^{\upalpha} \left(\frac{1+2\upalpha}{4\upalpha}\right)^{4\upalpha}  \right)^n$. When $0 \leq \upalpha \leq \frac{1}{6}$, the base achieves its maximum for $\upalpha = \frac{1}{6}$, at a value below $1.9064$. When $\upalpha > \frac{1}{6}$, we can switch back to the large matching case, with a bound $(2 \cdot (\frac{3}{4})^\upalpha)^n$ on the number of downsets. Since the two quantities are equal when $\upalpha=1/6$, and the second quantity is decreasing in $\upalpha$, we obtain the overall bound $O(1.9064^n)$ %
on the running time. 

To reach the result given in Theorem~\ref{thm1}, we need further ideas. Let us first prove the two lemmas from which the correctness of the current algorithm and its analysis follow. 

\begin{proof}[Proof of Lemma~\ref{lem:struct3}]
Let $q = \prod_{i=1}^\ell {({a_i}!)}$.
We describe an explicit mapping from linear extensions of $\ccP$ to linear extensions of $\ccP'$. 

Consider a linear extension $<$ of $\ccP$ viewed as a sequence $z = (z_1, \dots, z_n)$, where $z_1 < \dots < z_n$. The sequence $z$ contains $\ell$ disjoint subsequences of lengths $a_1,\dots,a_{\ell}$ formed respectively by the elements of $A_1, \dots, A_\ell$. 
Let $z' = (z'_1, \dots, z'_n)$ be the sequence obtained from $z$ by replacing, for all $i \in [\ell]$, the elements of $A_i$ in the sequence $z$, in the order of their appearance, by the virtual elements $x_i^1, \dots, x_i^{a_i}$. 

We proceed via two claims about the mapping $z \rightarrow z'$ from which the statement follows: (1) $z'$ is a linear extension of $\ccP'$, and (2) for every linear extension $z'$ of $\ccP'$ there are $q$ different linear extensions of $\ccP$ that map to $z'$. 

For (1), let $i_1, i_2$ be two arbitrary indices $1 \leq i_1 < i_2 \leq n$. We need to show that $z'_{i_2} \nprec' z'_{i_1}$. The four cases to consider are: (1a) $z'_{i_1}, z'_{i_2} \in W$, (1b) $z'_{i_1} = x_i^p$ and $z'_{i_2} = x_j^q$, (1c) $z'_{i_1} \in W$ and $z'_{i_2} = x_i^p$, and (1d) $z'_{i_1} = x_i^p$, and $z'_{i_2} \in W$. These correspond to the four cases in the definition of $\prec'$ and the claim easily follows in each case by the construction of $z'$.

For (2), consider a linear extension (sequence) $z'$ of $\ccP'$, and for all $i \in [\ell]$ replace the elements $\{x_i^1,\dots,x_i^{a_i}\}$ in $z'$ by an arbitrary permutation of the elements of $A_i$. In this way we obtain $q$ different linear extensions of $\ccP$, and when applying the above mapping to these linear extensions, they all yield the same $z'$. 
\end{proof}

\begin{proof}[Proof of Lemma~\ref{lem:struct2}]
Let $t = |A|$, and let us label the elements of $A$ as $z_1, \dots, z_t$. Let $<_1$ and $<_2$ be the realizers of the two-dimensional poset $\ccP$. Then, as $A$ is an antichain, its elements can be labeled such that $z_1 <_1 \cdots <_1 z_t$, and $z_t <_2 \dots <_2 z_1$. 
The crucial observation is that the neighborhood of an arbitrary $y \in X \setminus A$ in $A$ is defined by an \emph{interval} of indices. 

Formally, for $y \in X \setminus A$ that is above or below $A$, let $z_i, z_j$ be the elements of $N(y) \cap A$ with smallest, resp.\ largest index (it may happen that $i=j$). Define $b(y) = i-0.5$ and $b'(y) = j+0.5$ the \emph{boundaries} of the neighborhood of $y$. If $y$ is incomparable with $A$, set the boundaries to dummy values $b(y) = 0$, $b'(y) = t+1$.

If $y$ is above $A$, then for all $k$ such that $b(y) < k  < b'(y)$, we have $z_k \prec y$. To see this, observe that $z_k <_1 z_j <_1 y$, and $z_k <_2 z_i <_2 y$.

Symmetrically, if $y$ is below $A$, then for all $k$ such that $b(y) < k  < b'(y)$, we have $y \prec z_k$. To see this, observe that $y <_1 z_i <_1 z_k$, and $y <_2 z_j <_2 z_k$.

Let $b_1, \dots, b_{2(n-t)}$ be the multiset of neighborhood boundaries sorted in increasing order. Their number is $2(n-t)$ as each of the $n-t$ elements of $X \setminus A$ contribute exactly two boundaries. Let us add the two dummy boundaries $b_0 = 0$ and $b_{2(n-t)+1} = t+1$ (in case they never occurred during the process).

The classes of $A$ defined by the partition $N(\cdot)$ are then 
of the form $\{z_j : b_{i} < j < b_{i+1}\}$  where $0 \leq i \leq 2(n-t)$. There are at most $2(n-t)+1$ such classes (not all boundaries are necessarily distinct, and we can now remove empty classes due to duplicate boundaries). Moreover, the two classes delimited by $b_0$ to the left, respectively by $b_{2(n-t)+1}$ to the right are identical, corresponding to elements of $A$ incomparable to all $y \in X \setminus A$. 
The claimed bound on the maximum number of classes follows. 
\end{proof}

\subsection{A faster algorithm}
We now describe the improvements to Algorithm~\#LE-2D and its analysis that lead to the running time claimed in Theorem~\ref{thm1}.

\subparagraph*{Canonical matchings.}

Observe that set $A$ in Lemma~\ref{lem:struct2} denotes an arbitrary antichain. When $A$ is assumed to be the complement of a maximum matching with a certain property, a stronger statement can be made about the neighborhood-partition of $A$. 

Let $M$ be a maximum matching of $\C(\ccP)$, let $W$ be its vertex set, and let $A = X \setminus W$.  
We call an edge $\{x_i,y_i\}$ of $M$ \emph{separated}, if there exist $x_1,x_2 \in A$ such that $x_i \prec x_1$ and $x_2 \prec y_i$. (In other words, $x_i$ is below $A$, and $y_i$ is above $A$.) Observe that, in this case, $x_1$ and $x_2$ must be the same, as otherwise $M$ could be made larger by replacing edge $\{x_i,x_j\}$ by the two edges $\{x_i,x_1\}$, $\{x_2,x_j\}$. A matching is \emph{canonical} if it contains no separated edges. 

We argue that in an arbitrary poset a canonical matching of the same size as the maximum matching can be found in polynomial time. Indeed, start with an arbitrary maximum matching $M$.
If $M$ contains no separated edges, we are done. Otherwise, let $\{x_i, y_i\}$ be an edge of $M$ with $x \in A$ such that $x_i \prec x \prec y_i$. (Such a triplet can easily be found in polynomial time.) 
Replace the edge $\{x_i, y_i\}$ in $M$ by the edge $\{x,y_i\}$. As $x$ was previously not matched, the resulting set of edges is still a maximum matching. We claim that with $O(n^2)$ such swaps we obtain a canonical matching (i.e.\ one without separated edges). To see this, consider as potential function the sum of ranks of all vertices in the current matching, according to an arbitrary fixed linear extension of $\ccP$. Each swap increases the potential by at least one (since $x$ must come after $x_i$ in every linear extension). Since the sum of ranks is an integer in $O(n^2)$, the number of swaps until we are done is also in $O(n^2)$. In the following, we can therefore assume that $M$ is a canonical maximum matching. %
We can now state the stronger structural lemma. 

\begin{lemma}\label{lem:struct4}
Let $\ccP = (X,\prec)$ be a size-$n$ poset, with $\dim(\ccP) \leq 2$. Let $M$ be a canonical maximum matching in $\C(\ccP)$ with vertex set $W$, and let $A = X \setminus W$. Then, $N(\cdot)$ partitions $A$ into at most $|W|$ classes. 
\end{lemma}

\begin{proof}
Since $M$ is canonical, for every edge $\{x_i,y_i\}$, one of the following must hold:
\begin{enumerate}[(i)]
\item $x_i$ and $y_i$ are both above $A$,
\item $x_i$ and $y_i$ are both below $A$,
\item $x_i$ is incomparable with $A$ and $y_i$ is above $A$,
\item $y_i$ is incomparable with $A$ and $x_i$ is below $A$.
\end{enumerate}

Recall that in the proof of Lemma~\ref{lem:struct2}, we considered, for all $y \in X \setminus A$, the two boundaries of the interval $N(y) \cap A$. Now, in cases (iii) and (iv), only one of $x_i$ and $y_i$ need to be considered, as the neighborhood of the other is disjoint from $A$.

In cases (i) and (ii), it is also sufficient to consider only one of $x_i$ and $y_i$ as we have $N(x_i) \cap A = N(y_i) \cap A$. Furthermore, in this case $|N(x_i) \cap A| = 1$. To see this, suppose that there are $z,z' \in A$ such that $z \in N(x_i)$ and $z' \in N(y_i)$. Then $M$ could be extended by replacing the edge $\{x_i,y_i\}$ with the edges $\{z,x_i\}$ and $\{z',y_i\}$, contradicting the maximality of $M$. 

It follows that in the argument of Lemma~\ref{lem:struct2} we only need to consider the intervals created by $|M|$ elements of $X \setminus A$, yielding the bound $2|M| = |W|$ on the number of classes. It is easy to construct examples where the bound is tight.
\end{proof}

It follows that, if we require the matching $M$ in Step~1 of Algorithm~\#LE-2D to be canonical, then by Lemma~\ref{lem:struct3}, the bound on the number of downsets reduces to 
$\left( 3^{\upalpha} \left(\frac{1}{2\upalpha}\right)^{2\upalpha}  \right)^n$. When $0 \leq \upalpha \leq \frac{1}{4}$, the base achieves its maximum for $\upalpha = \frac{1}{4}$, at a value below $1.8613$. When $\upalpha > \frac{1}{4}$, we can switch back to the large matching case, with a bound $(2 \cdot (\frac{3}{4})^\upalpha)^n$. Since the two quantities are equal when $\upalpha=1/4$, and the second quantity is decreasing in $\upalpha$, we obtain the overall bound $O(1.8613^n)$ %
on the running time.

\subparagraph*{Packing triplets and quartets.} 
The final improvement in running time comes from the attempt to find, instead of a matching (i.e.\ a packing of edges), a packing of larger connected structures. %
Beyond the concrete improvement, the technique may be of more general applicability and interest, which we illustrate in §\,\ref{sec3} for the jump number problem.

Assume, as before, that $M$ is a canonical maximum matching of $\C(\ccP)$ of size $\upalpha n$ with vertex set $W$ and that $A$ denotes the antichain $X \setminus W$. %
Let us form an auxiliary bipartite graph $B$ with vertex sets $L$ and $R$, where $L = A$, and $R = M$, i.e.\ $R$ consists of the \emph{edges} of $M$. A vertex $x \in L$ is connected to a vertex $\{x_i,y_i\} \in R$ exactly if $x$ is comparable to one or both of $x_i$ and $y_i$. Let $M_B$ be a maximum matching of $B$, of size $\upbeta n$. Clearly, $\upbeta \leq \upalpha$.

Edges of $M_B$ connect vertices in $A$ to matched edges of $M$, forming \emph{triplets} of vertices of $X$ that induce connected subgraphs in $\C(\ccP)$. Let $T$ denote the set of all triplets created by edges of $M_B$. 

Let us form now another auxiliary bipartite graph $B'$ with vertex sets $L'$ and $R'$, where $L'$ consist of the vertices of $A$ \emph{unmatched} in $M_B$, and $R' = T$, i.e.\ the triplets found in the previous round. A vertex $x \in L'$ is connected to a vertex $z \in R'$ exactly if $x$ is comparable to at least one of the vertices forming the triplet $z$. Let $M_{B'}$ be a maximum matching of $B'$, and denote its size by $\upgamma n$. Clearly, $\upgamma \leq \upbeta$. 

Edges of $M_{B'}$ connect vertices in $A$ to triplets of $T$, forming \emph{quartets} of vertices of $X$ that induce connected subgraphs in $\C(\ccP)$. Let $Q$ denote the set of all quartets created by edges of $M_{B'}$.

Let $A'$ denote the vertices of $A$ that were not matched in either of the two matching rounds. Observe that $|A'| = n(1-2\upalpha - \upbeta - \upgamma)$. We make the following observations.

(1) The endpoints of edges of $M$ that were \emph{unmatched} in $M_B$ are not comparable to any vertex in $A'$, as otherwise $M_B$ would not have been maximal. %
There are ${n(\upalpha-\upbeta)}$ such unmatched edges. These contribute a factor of $3^{n(\upalpha-\upbeta)}$ to the number of downsets. 

(2) The vertices in triplets of $T$ that were \emph{unmatched} in $M_{B'}$ are not comparable to any vertex in $A'$, as otherwise $M_{B'}$ would not have been maximal. There are $n(\upbeta - \upgamma)$ such triplets. A simple case-analysis shows that the number of downsets of a size-$3$ poset with connected comparability graph is at most $5$. It follows that these triplets contribute a factor of at most $5^{n(\upbeta-\upgamma)}$ to the number of downsets.

(3) There are $\upgamma n$ quartets in $Q$. A case-analysis\footnote{An easy induction shows more generally that the maximum number of downsets of a size-$n$ poset with connected comparability graph is $2^{n-1}+1$.} shows that the number of downsets of a size-$4$ poset with connected comparability graph is at most $9$. It follows that these quartets contribute a factor of at most $9^{n \upgamma}$ to the number of downsets.

(4) All vertices in $X \setminus A'$ are accounted for. As for the vertices in $A'$, we partition them into classes $A_1, \dots, A_\ell$ by $N(\cdot)$, and apply the same transformation as previously, creating a new poset $\ccP'$. By the previous discussion, only the vertices from the quartets in $Q$ may be comparable to vertices in $A'$. Furthermore, in each quartet, only the vertices coming from the original matching $M$ may be comparable to a vertex in $A'$ (other vertices come from the antichain $A \supseteq A'$). Thus, by Lemma~\ref{lem:struct4}, the number of classes created on $A'$ is $\ell \leq 2\upgamma n$.

Putting everything together, assuming $\upgamma >0$ (the case $\upgamma=0$ is discussed later),
we obtain the upper bound $\tau^n$ on the number of downsets of $\ccP'$, where 
$$\tau = \tau(\upalpha,\upbeta,\upgamma) = 3^{(\upalpha-\upbeta)} \cdot 5^{(\upbeta-\upgamma)} \cdot 9^{\upgamma} \cdot \left(\frac{1-2\upalpha-\upbeta-\upgamma}{2\upgamma} + 1\right)^{2\upgamma}.$$

As an alternative, consider the number of downsets of the original poset $\ccP$. We account for the factors contributed by edges of $M$ unmatched in $M_B$, triplets in $T$ unmatched in $M_{B'}$ and quartets in $T$. For $A'$, however, we consider all possible subsets (not partitioning $A'$ into classes). We obtain the upper bound $\pi^n$, where $$\pi = \pi(\upalpha,\upbeta,\upgamma) = 3^{(\upalpha-\upbeta)} \cdot 5^{(\upbeta-\upgamma)} \cdot 9^{\upgamma} \cdot 2^{1-2\upalpha - \upbeta - \upgamma} = 2 \cdot \left( \frac{3}{4} \right)^{\upalpha} \cdot \left( \frac{5}{6} \right)^{\upbeta} \cdot \left( \frac{9}{10} \right)^{\upgamma}.$$ 

We distinguish two cases:

\noindent\textbf{(i)} $2\upgamma \geq 1-2\upalpha - \upbeta - \upgamma$. Then, $\pi \leq \tau$, and it is advantageous to use the original poset $\ccP$. 

\noindent\textbf{(ii)} $2\upgamma \leq 1-2\upalpha - \upbeta - \upgamma$. Then $\pi \geq \tau$, and it is advantageous to use the transformed $\ccP'$. 

In case of equality, we choose arbitrarily between the two cases. Intuitively, the inequality of case (i) means that $A'$ may even be partitioned into $|A'|$ classes, in which case we do not gain from partitioning at all. In both cases we have the additional constraints $0 < \upgamma \leq \upbeta \leq \upalpha < \frac{1}{3}$, and $1-2\upalpha-\upbeta-\upgamma \geq 0$.

Under these constraints we obtain that $\pi < 1.8172$ in case (i), respectively, $\tau < 1.8172$ in case (ii). The maximum is attained in both cases for $\upalpha = \upbeta = \upgamma = \frac{1}{6}$. Observe that the least favorable situation is when all edges of $M$ are matched into triplets, and all triplets are matched into quartets. 

When $\upgamma = 0$, no quartets are created, and $A'$ forms a single class, transformed in $\ccP'$ into a single chain, contributing a linear factor to the overall bound. Thus, the upper bound $n \cdot \tau^n$ holds, with $\tau = \tau(\upalpha,\upbeta) = 3^{(\upalpha-\upbeta)} \cdot 5^\upbeta < 1.71$, maximum attained for $\upalpha = \upbeta (\approx 1/3)$.  

The resulting algorithm~\#LE-2D$^{\ast}$ is listed as Algorithm~\ref{alg2}. The correctness and running time bounds (Theorem~\ref{thm1}) follow from the previous discussion. We defer some remarks about the algorithm and its analysis to §\,\ref{secrem}.

\subparagraph*{Open questions.} The following questions about counting linear extensions are suggested in increasing order of difficulty. (1) Can \#LE be solved in two-dimensional posets faster than the algorithm of Theorem~\ref{thm1}? (2) Can \#LE be solved in time $O(c^n)$ for $c<2$ in $d$-dimensional posets, for $d \geq 3$? (3) Can \#LE be solved in time $O(c^n)$ for $c < 2$ in arbitrary posets?

\newpage

\section{The jump number problem}
\label{sec3}

In this section we present our improvement for the jump number problem. We start with a formal definition of the problem, and the straightforward dynamic programming. We then review the algorithm of Kratsch and Kratsch, followed by our extension. The result is intended as an illustration of the matching technique of §\,\ref{sec2}, which is not specific to two-dimensional posets.

Given a linear extension $x_1 < \dots <x_n$ of a poset $\ccP = (X,\prec)$, a pair of neighbors $(x_i$, $x_{i+1})$ is a \emph{jump} if $x_i \nprec x_{i+1}$, and is a \emph{bump} if $x_i \prec x_{i+1}$. The number of jumps, resp.\ bumps of the linear extension $<$ of $\ccP$ is denoted as $\jump(<)$, resp.\ $\bump(<)$. The jump number problem asks to compute the minimum possible value $\jump(<)$ for a linear extension $<$ of $\ccP$. Additionally, a linear extension realizing this value should be constructed. In the algorithms we describe, obtaining a linear extension that realizes the minimum jump number is a mere technicality, we thus focus only on computing the minimum jump number. 

An easy observation is that the relation $\jump(<) + \bump(<) = n-1$ holds for all linear extensions $<$ of $\ccP$. Minimizing the number of jumps is thus equivalent to maximizing the number of bumps, allowing us to focus on the latter problem. 

Let $\bump(\ccP)$ denote the maximum bump number of a linear extension of $\ccP$. For all $Y \subseteq X$, and $x \in \max{(Y)}$, let $\bump(Y,x)$ denote the maximum bump number of a linear extension of the subposet of $\ccP$ induced by $Y$ that ends with element $x$. Let us define $\bump(\{x\},x) = 0$, for all $x \in X$. We recursively express $\bump(Y,x)$ by removing $x$ from the end and trying all remaining elements in turn as the new last element: 
\begin{align}
\bump(Y,x) = \max_{y \in \max{(Y \setminus \{x\})}} {\Bigl(  \bump(Y \setminus \{x\}, y) + \left[y \prec x\right] \Bigr)}.\label{rec2}
\end{align}

The term $[y \prec x]$ denotes the value $1$ if $y \prec x$, i.e.\ if the last pair forms a bump, and $0$ otherwise. Executing recurrence (\ref{rec2}) na\"ively leads to an algorithm that computes $\bump(\ccP)$ in time $O(2^n \cdot n^2)$.

We now describe the improvement of Kratsch and Kratsch~\cite{jumpnr}. Observe that jumps partition a linear extension of $\ccP$ uniquely into a sequence of chains of $\ccP$, such that the last element of each chain is incomparable with the first element of the next chain, and all other neighboring pairs are comparable.

Consider a linear extension with minimum jump number and let $C_1, \dots, C_k$ denote the non-trivial chains of its decomposition (i.e.\ all chains of length at least $2$). Let $C$ denote the set of vertices of chains $C_1, \dots, C_k$. Then, as all bumps occur between elements of $C$, the bump number of $\ccP$ equals the bump number of the subposet induced by $C$. In other words, to compute the maximum bump number, it is sufficient to consider in recurrence (\ref{rec2}) the subsets of the ground set $X$ that are candidate sets $C$ in the optimum.

Kratsch and Kratsch consider a maximum matching $M$ of $\C(\ccP)$ with vertex set $W$ and observe that the vertices of the antichain $A = X \setminus W$ that participate in nontrivial chains (i.e.\ that are in $C$) form a matching with vertices of $W$. (This is because a vertex $v \in A$ can only form a bump together with a vertex from $X \setminus A$, and two vertices $v,v' \in A$ cannot form bumps with the same vertex, as that would contradict their incomparability.) Moreover, $v,v' \in A$ cannot be the neighbors of the two endpoints of a matched edge of $M$, as that would contradict the maximality of $M$.

Thus, it suffices to compute (\ref{rec2}) over subsets of $X$ that consist of $W \cup A'$, where $A' \subseteq A$ and $|A'| \leq |M|$. Furthermore, only \emph{downsets} of $\ccP$ need to be considered, leading to a further saving due to the fact that $W$ forms a matching. Denoting $|M| = \upalpha n$, the overall number of subsets of $X$ that need to be considered is ${(1-2\upalpha)n \choose {\leq \upalpha n}} \cdot 3^{\upalpha n}$. %

\subparagraph*{Packing triplets.} 

We now describe our improvement. Again, let $M$ be a canonical maximum matching of $\C(\ccP)$ of size $\upalpha n$ with vertex set $W$ and let $A$ denote the antichain $X \setminus W$. Form an auxiliary bipartite graph $B$ with vertex sets $L$ and $R$, where $L = A$, and $R = M$. A vertex $x \in L$  is connected to a vertex $\{x_i,y_i\} \in R$ (i.e.\ an edge of $M$) exactly if $x \prec y_i$ or $x_i \prec x$. Let $M_B$ be a maximum matching of $B$, and denote its size by $\upbeta n$. Clearly, $\upbeta \leq \upalpha$.

Edges of $M_B$ connect vertices in $A$ to matched edges of $M$, forming \emph{triplets} of vertices of $X$ that induce connected subgraphs in $\C(\ccP)$. Let $T$ denote the set of all such triplets. (To keep the argument simple we forgo in this case further rounds of matching and the forming of quartets. The result is thus not optimized to the fullest extent.) 

Let $A'$ denote the vertices of $A$ that were not matched in $M_B$. Observe that $|A'| = n(1-2\upalpha - \upbeta)$. We make the following observations.

(1) The endpoints of edges of $M$ that were \emph{unmatched} in $M_B$ are not comparable to any vertex in $A'$ (assuming that $A'$ is nonempty), as otherwise $M_B$ would not have been maximal. %
There are ${n(\upalpha-\upbeta)}$ such edges. These edges contribute a factor of $3^{n(\upalpha-\upbeta)}$ to the number of downsets. 

(2) There are $\upbeta n$ triplets in $T$. These contribute a factor of at most $5^{n \upbeta}$ to the number of downsets.

(3) All vertices in $X \setminus A'$ are accounted for. Vertices of $A'$ that participate in non-trivial chains of the optimal linear extension can be matched to vertices in the triplets of $T$. A vertex $v \in A'$ can only be connected to those vertices of a triplet in $T$ that are endpoints of an edge in $M$ (all other vertices come from $A$, are thus incomparable with $v$). Furthermore, $v,v' \in A'$ may not connect to different endpoints of the same edge in $M$, as that would contradict the maximality of $M$. It follows that each vertex in $A'$ that participates in $C$ must be matched to a unique triplet in $T$. Thus, at most $\upbeta n$ vertices of $A'$ need to be considered.

The resulting Algorithm~JN is listed as Algorithm~\ref{alg3}. Its correctness follows from the previous discussion.

\subparagraph*{Running time.}
In the large matching ($\upalpha \geq \frac{1}{3}$) case, the bound given in §\,\ref{sec2} on the number of downsets holds, and the running time is within the bound of Theorem~\ref{thm3}. We assume therefore that $\upalpha < \frac{1}{3}$.

In the special case $\upbeta = 0$ only vertices in $M$ need to be considered, with an overall upper bound $3^{\upalpha n}$ on the number of downsets. For $\upalpha < \frac{1}{3}$, this quantity is below $1.443^n$.

Assuming $\upbeta > 0$, %
we have an upper bound $\tau^n$ on the number of downsets of $\ccP$, where 
$\tau^n = 3^{(\upalpha-\upbeta)n} \cdot 5^{\upbeta n} \cdot {{n(1-2\upalpha-\upbeta)} \choose {\leq \upbeta n}}$.
To obtain a simpler expression, we use a standard upper bound~\cite[p.~406]{cover2006} on the sum of binomial coefficients. Assuming $0 \leq 2b \leq a \leq 1$, we have 
${n a \choose {\leq n b}} = \sum_{k=0}^{n b}{{n a} \choose  k} \leq  n^{O(1)} \cdot \left(\frac{a^a}{b^b \cdot (a-b)^{(a-b)}}\right)^n$.

Plugging in $a = 1-2\upalpha -\upbeta$ and $b = \upbeta$, and assuming $2b \leq a$, we have $2\upalpha + 3\upbeta \leq 1$. Omitting the polynomial factor, we obtain $\tau \leq 3^{(\upalpha-\upbeta)} \cdot 5^{\upbeta} \cdot \frac{(1-2\upalpha-\upbeta)^{(1-2\upalpha-\upbeta)}}{\upbeta^\upbeta \cdot (1-2\upalpha-2\upbeta)^{(1-2\upalpha-2\upbeta)}}$. In the critical region $0 < \upbeta \leq \upalpha < \frac{1}{3}$ we obtain the bound $\tau < 1.824$, with the maximum attained for $\upalpha = \upbeta  (\approx 0.1918)$. %

When $2\upalpha + 3\upbeta \geq 1$, we use the easier upper bound on the sum of binomial coefficients ${n a \choose {\leq n b}} \leq 2^{n a}$, obtaining $\tau \leq 3^{(\upalpha-\upbeta)} \cdot 5^{\upbeta} \cdot 2^{(1-2\upalpha-\upbeta)}$. In the allowed range $0 < \upbeta \leq \upalpha < \frac{1}{3}$ and additionally requiring $2\upalpha + 3\upbeta \geq 1$, the quantity is maximized for $\upalpha = \upbeta = 0.2$, yielding $\tau < 1.8206$, within the required bounds.

\section{Discussion} \label{secrem}

We start with some remarks about algorithm \#LE-2D$^\ast$. It is straightforward to extend this algorithm beyond pairs, triplets, and quartets, to also form $k$-tuples for $k > 4$ via further matching rounds. A similar analysis, however, indicates no further improvements in the upper bound. When forming triplets and quartets, other strategies are also possible. For instance, we may try to combine connected pairs of edges from $M$ into quartets. The quartets formed in this way are, in fact, preferable to those obtained by augmenting triplets, as their number of downsets is strictly less than the value $9$ given before. (The value $9$ is attained when $3$ of the $4$ vertices form an antichain, which is not possible if the quartet consists of two matched edges.)  

We observe that in some instances, the largest antichain may be significantly larger than the antichain $A$ obtained as the complement of the maximum matching. One can find the largest antichain in time $O(n^{5/2})$ via a reduction to bipartite matching (see e.g.\ \cite{Tarjan_book}). In these cases, using the partition of the largest antichain may lead to a better running time, even if its complement does not admit a perfect matching. %
In two-dimensional posets with a realization $<_1, <_2$, the largest antichain can be found in time $O(n\log{n})$ by reduction to the \emph{largest decreasing subsequence} problem. In our analysis, we assumed the classes $A_i$ to be of equal size. The running time can, of course, be significantly lower when the distribution of class sizes is far from uniform.

We further remark that the actual time and space requirement of our algorithms is dominated by the number of downsets of a given poset $\ccP$ (or the transformed poset $\ccP'$).  The number of downsets (order ideals) is known to equal the number of antichains~\cite[§\,3]{Stanley_book}. Counting antichains is, in general, $\#\P$-hard~\cite{ProvanBall}, but solvable in two-dimensional posets in polynomial time~\cite{Steiner93, Mohring}. 
Thus, assuming that the transformed poset $\ccP'$ is also two-dimensional, we can efficiently compute a precise, \emph{instance-specific} estimate of the time and space requirements of our algorithms. 

To see that indeed, $\dim(\ccP') \leq 2$, recall that $\ccP'$ is obtained from $\ccP$ by replacing antichains $A_i$ by chains of equal size, such as to preserve the comparability of the involved elements to elements in $X \setminus A$. We can obtain a two-dimensional embedding of $\ccP'$ by starting with a two-dimensional embedding of $\ccP$, with points having integer coordinates, and no two points aligned on either coordinate. For an arbitrary $x \in A_i$, form a $0.5 \times 0.5$ box around the point $x$, and place the chain replacing $A_i$ on the main diagonal of this box. Then the comparability of points in the chain with elements in $X \setminus A_i$ is the same as for the point $x$.

\subparagraph*{Optimization.} Our first two bounds in §\,\ref{sec2} depend only on the fraction $\upalpha$ of matched vertices, and their maxima are found using standard calculus. The final bounds given in Theorem~\ref{thm1} and Theorem~\ref{thm3} however, require us to optimize over unwieldy multivariate quantities with constrained variables. The given numerical bounds were obtained using Wolfram Mathematica software. We have, however, independently certified the bounds, by the method illustrated next.

Let $\tau = \tau(\upalpha,\upbeta,\upgamma) = 3^{\upalpha} \cdot \left(\frac{5}{3}\right)^{\upbeta} \cdot {\left( \frac{9}{5} \right)}^{\upgamma} \cdot \left(\frac{1-2\upalpha-\upbeta-\upgamma}{2\upgamma} + 1\right)^{2\upgamma}$, and some value $\bar{\tau}>0$. Suppose we want to show that $\tau < \bar{\tau}$ for all $A \leq \upgamma \leq \upbeta \leq \upalpha \leq B$, where $A,B \in (0,\frac{1}{3})$, and $2\upalpha + \upbeta + \upgamma \leq 1$.

Consider a box $\mathcal{B} = [\upalpha_1, \upalpha_2] \times [\upbeta_1, \upbeta_2] \times [\upgamma_1,\upgamma_2] \subseteq [A,B]^3$. Then, at an \emph{arbitrary} point $(\upalpha, \upbeta, \upgamma) \in \mathcal{B}$, the following upper bound holds.
$$\tau(\upalpha,\upbeta,\upgamma) \leq 3^{\upalpha_2} \cdot {\left( \frac{5}{3} \right) }^{\upbeta_2} \cdot {\left( \frac{9}{5} \right)}^{\upgamma_2} \cdot \left(\frac{1-2\upalpha_1-\upbeta_1-\upgamma_1}{2\upgamma_2} + 1\right)^{2\upgamma_2}. $$

To show $\tau < \bar{\tau}$, it is sufficient to exhibit a collection of boxes, such that (1) for all boxes, the stated upper bound evaluates to a value smaller than $\bar{\tau}$, and (2) the union of the boxes covers the entire domain of the variables. We can find such a collection of boxes if we start with a single box that contains the entire domain of the variables, and recursively split boxes into two equal parts (along the longest side) whenever the upper bound evaluates to a value larger than the required value.

\subparagraph*{Higher dimensions.}
A straightforward extension of Algorithms~\#LE-2D and \#LE-2D$^{\ast}$ to higher dimensional posets does not yield improvements over the na\"ive dynamic programming. The crux of the argument in two dimensions is that a large antichain in $\ccP$ is split, according to the neighborhoods in $\C(\ccP)$ into a small number of classes. In dimensions three and above it is easy to construct posets with an antichain containing \emph{almost all} elements, e.g.\ such that $|X \setminus A| = O(\sqrt{n})$, with the property that all elements of $A$ have unique neighborhoods. In this case, the number of classes is $|A|$ and the described techniques yield no significant savings.

\newpage

\begin{algorithm}
  \caption{Algorithm~\#LE-2D$^{\ast}$}\label{alg2}
  \begin{algorithmic}[1]
    \Statex \textbf{Input:} Poset $\ccP = (X,\prec)$, where $|X| = n$.
    \Statex \textbf{Output:} The number of linear extensions $\LE(\ccP)$ of $\ccP$. 
    \State Find a maximum matching $M$ of $\C(\ccP)$ with vertex set $W$. 
    \State Let $A = X \setminus W$.
    \State Find $T$ and $Q$ as described, and let $A'$ be the unmatched part of $A$.
    \State Let $A_1, \dots, A_\ell$ be the partition of $A'$ by the neighborhoods in $\C(\ccP)$.
    \State Let $a_i = |A_i|$ for $i \in [\ell]$.
    \State Construct $\ccP' = (X',\prec')$.
    \State Compute $N = \LE(\ccP')$  using (\ref{rec1}). %
\State \textbf{return} $N \cdot \prod_{i\in[\ell]}{(a_i!)}$.
    \end{algorithmic}
\end{algorithm}

\begin{algorithm}
  \caption{Algorithm~JN}\label{alg3}
  \begin{algorithmic}[1]
    \Statex \textbf{Input:} Poset $\ccP = (X,\prec)$, where $|X| = n$.
    \Statex \textbf{Output:} The minimum jump number of a linear extension of $\ccP$. 
    \State Find a maximum matching $M$ of $\C(\ccP)$ with vertex set $W$. 
    \State Let $A = X \setminus W$.
    \State Find $T$ as described, and let $A'$ be the unmatched part of $A$.
    \State Let $\upbeta = |T|$.
    \State Compute $B = \bump(\ccP)$ by (\ref{rec2}), using downsets of $\ccP$ with at most $\upbeta n$ vertices from $A'$.  %
    \State \textbf{return} $n-1- B$.
    \end{algorithmic}
\end{algorithm}

\clearpage
\bibliography{submission}

\end{document}